\documentclass[reprint,superscriptaddress,amsmath,amssymb,aps,prl,floatfix]{revtex4-2}
\usepackage{graphicx} 
\usepackage{xcolor}
\usepackage{etoolbox}
\usepackage{physics}
\usepackage{bbm}
\usepackage{amsthm}
\usepackage{hyperref}
\hypersetup{
    colorlinks,
    citecolor=blue,
    filecolor=blue,
    linkcolor=blue,
    urlcolor=blue
}

\newtheorem{theorem}{Theorem}
\newtheorem{theorem2}{Theorem}
\newtheorem*{corollary}{Corollary}
\newtheorem*{corollary2}{Corollary}

\makeatletter
\def\maketitle{
\@author@finish
\title@column\titleblock@produce
\suppressfloats[t]}
\makeatother

\allowdisplaybreaks
\begin{document}

\title{Certifiable Lower Bounds of Wigner Negativity Volume and \texorpdfstring{\\}{}Non-Gaussian Entanglement with Conditional Displacement Gates}

\author{Lin Htoo Zaw}
\email[Electronic mail: ]{htoo@zaw.li}
\affiliation{Centre for Quantum Technologies, National University of Singapore, 3 Science Drive 2, Singapore 117543}

\begin{abstract}
In circuit and cavity quantum electrodynamics devices where control qubits are dispersively coupled to high-quality-factor cavities, characteristic functions of cavity states can be directly probed with conditional displacement (CD) gates. In this Letter, I propose a method to certify non-Gaussian entanglement between cavities using only CD gates and qubit readouts. The CD witness arises from an application of Bochner's theorem to a surprising connection between two negativities: that of the reduced Wigner function, and that of the partial transpose. Non-Gaussian entanglement of some common states, like entangled cats and photon-subtracted two-mode squeezed vacua, can be detected by measuring as few as four points of the characteristic function. Furthermore, the expectation value of the witness is a simultaneous lower bound to the Wigner negativity volume and a geometric measure of entanglement conjectured to be the partial transpose negativity. Both negativities are strong monotones of non-Gaussianity and entanglement, respectively, so the CD witness provides experimentally accessible lower bounds to quantities related to these monotones without the need for tomography on the cavity states.
\end{abstract}

\phantomsection\addcontentsline{toc}{part}{Certifiable Lower Bounds of Wigner Negativity Volume and Non-Gaussian Entanglement with Conditional Displacement Gates}
\maketitle

\section{Introduction}
Recent advances in both circuit and cavity quantum electrodynamics (cQED) have led to a paradigm where continuous-variable computation is performed in high-quality-factor cavities, in which operations on the quantum states of the cavities are mediated by qubits dispersively coupled to them \cite{cQED5,cQED4,cQED3,cQED2,cQED6,cQED7}. In systems with weak dispersive coupling, measurements with conditional displacement (CD) gates can directly probe the pointwise characteristic function of cavity states \cite{ionCDgate,ECDgate,CNODgate}. Meanwhile, displaced parity gates are possible but would take longer times, while quadrature measurements are unavailable.

Many existing continuous-variable entanglement witnesses are based on quadrature statistics \cite{witness-quad-simon,witness-quad-duan,witness-quad-all,witness-quad-three,witness-quad-multi} or require nonlocal photon-counting measurements \cite{witness-photon-count-1,witness-photon-count-2}. As such, past demonstrations of entanglement in such architectures have instead resorted to violating Bell inequalities with Wigner function measurements \cite{EntBell1,EntangledCats,EntBell2} or by computing the entanglement fidelity from the tomographically reconstructed state \cite{EntFidelity1,EntFidelity2,EntFidelity3,CNODgate}. The former can be difficult in weakly coupled systems due to the necessity of parity gates, while the latter is an expensive operation. Instead, an entanglement witness that uses only a few instances of CD gates would be preferable in such setups.

Meanwhile, in resource-theoretic studies of non-Gauss\-ian\-i\-ty and entanglement, Wigner logarithmic negativity is a common measure that quantifies non-Gaussianity \cite{NonGaussianResourceTheory1,NonGaussianResourceTheory2}, while logarithmic partial transpose negativity is a common measure that quantifies entanglement \cite{LogNegativityMonotone}. Apart from the property that they are non-Gaussian and entanglement monotones respectively, they are also favored in theoretical studies for their computability: the former is an integral over negative regions of the Wigner function, while the latter is the sum of the singular values of a partially transposed matrix. Given a state description, both can be easily calculated with numerical tools.

However, these quantities are difficult to obtain in experimental settings, as full tomography is required to reconstruct the state for either negativities to be calculated. Therefore, it is often desirable to instead find lower bounds to these quantities with just a few experimentally accessible observables \cite{quantify-PPT-1,quantify-PPT-2}.

The contributions of this Letter are twofold. First, I prove that a Wigner negativity witness based on Bochner's theorem provides a lower bound to the Wigner negativity volume. Previous Wigner negativity witnesses have only been shown to lower bound the trace distance negativity, which has not been proven to be a non-Gaussian monotone \cite{CharacteristicFunctionWitness,WitnessingWignerNegativity}.

Second, this Letter introduces a method to directly detect non-Gaussian entanglement between cavities using only CD gates and qubit readouts. In most cases, as few as four settings of the CD gates are needed to certify non-Gaussian entanglement.

These contributions convene in the expected value of the proposed non-Gaussian entanglement witness, which simultaneously lower bounds the Wigner negativity volume and a geometric measure of entanglement. The latter is conjectured, and proven for a large family of states, to be equivalent to the partial transpose negativity.

\section{\label{sec:weakRegime}Weak-dispersive Regime and Conditional Displacement Gates}
In circuit QED (cavity QED) systems where qubits are dispersively coupled to driven cavity modes, readout lines (control lasers) are coupled only to the qubits \cite{cQED4,cQEDreview} (\cite{cQED3,cQED6,cQED7}), so only qubit measurements are available in both cases. Hence, entangling gates between qubits and cavities are required to infer some properties of the cavity states from the qubit measurements. Such entangling operations typically have gate times of order $\sim \chi^{-1}$, where $\chi$ is the cross-Kerr nonlinearity that dictates the coupling strength between the qubit and the cavity \cite{displacedParity}.

On the other hand, $\chi$ has to be small to suppress higher-order effects like the anharmonicity of the cavity \cite{cQED4,cQEDreview}. This results in a trade-off between suppression of nonlinear effects and entangling gate times.

Some recently developed control schemes---using spin-dependent forces \cite{ionCDgate}, echoed conditional displacements \cite{ECDgate}, or conditional-\textsc{not} displacements \cite{CNODgate}---offer alternatives to this trade-off. In all three schemes, the system nonlinearities are small, and an appropriate drive is chosen to induce the time evolution of the system to be a conditional displacement (CD) gate of the form \cite{ionCDgate,ECDgate,CNODgate}
\begin{equation}
    U_{\text{CD}}(\vec{\xi}) := D(\vec{\xi}/2)\ketbra{e}{g} + D(-\vec{\xi}/2)\ketbra{g}{e},
\end{equation}
where $D(\vec{\xi}) = {\exp}(\sum_{m=1}^M \xi_m a_m^\dag - \xi_m^* a_m) =: {\exp}(\vec{\xi} \wedge \vec{a})$ is the usual displacement operator on the cavity modes $\vec{a}$, with the shorthand $\vec{x} \wedge \vec{y} := \sum_k [\vec{x}]_k [\vec{y}]_k^\dag - [\vec{y}]_k [\vec{x}]_k^\dag$, and $\ket{g}$ ($\ket{e}$) is the ground (excited) state of the qubit.

If the qubit is initialized to $\ket{+} \propto \ket{g} + \ket{e}$ so that the initial state of the combined system is $\rho_I = \rho\otimes\ketbra{+}$, where $\rho$ is the cavity state, then performing a CD gate and measuring the qubit along the $\sigma_x = \ketbra{g}{e} + \ketbra{e}{g}$ or $\sigma_y = - i\ketbra{g}{e} + i\ketbra{e}{g}$ basis gives
\begin{equation}
\begin{aligned}
    \ev{\sigma_x} &= \tr[ \sigma_x U_{\text{CD}}(\vec{\xi}) \rho_I U^\dag_{\text{CD}}(\vec{\xi})] = \Re\!\Bqty{\tr[\rho D(\vec{\xi})]},\\
    \ev{\sigma_y} &=  \tr[ \sigma_y U_{\text{CD}}(\vec{\xi}) \rho_I U^\dag_{\text{CD}}(\vec{\xi})] 
    = \Im\!\Bqty{\tr[\rho D(\vec{\xi})]},
\end{aligned}
\end{equation}
from which ${\tr}[\rho D(\vec{\xi})] = \ev{\sigma_x} + i \ev{\sigma_y}$ can be obtained. This is the characteristic function of $\rho$, which is related to its Wigner function by a Fourier transform \cite{WignerReview}
\begin{equation}
    \tr[\rho D(\vec{\xi})] = \int\dd[2M]{\vec{\alpha}} 
    W_\rho(\vec{\alpha}) e^{\vec{\xi}\wedge\vec{\alpha}}.
\end{equation}
Implementations of CD gates have demonstrated their gate times to be orders of magnitude faster than the $\sim \chi^{-1}$ of displaced parity gates required for Wigner function measurements \cite{ionCDgate,ECDgate,CNODgate,displacedParity}.

Both characteristic and Wigner functions provide complete tomographic information of the cavity state, from which entanglement can be inferred. However, as full tomography is an expensive procedure, it is often desirable to certify entanglement from just a few measurements.

\section{\label{sec:prerequisiteTheorems}Theoretical Results}
The entanglement witness introduced in this Letter arises from a corollary of two theorems that are proven in Sec.~S1 of the Supplemental Material \cite{supplementaryMaterial}.%
\nocite{
        WignerNegativityVolume,
        wignerReduced,
        WignerTranspose,
        EPR-Bell,
        EPR-Unnormalized,
        MatrixAnalysis,
        MikeAndIke,
        GaussianQMreview,
        NegativeEnt1,
        Hudson1,
        Hudson2,
        PSTMSS,
        LogNegativityKnown}
As they might be of theoretical interest independent from their usefulness in constructing an entanglement witness, I shall briefly discuss the results of both theorems here.

The first is related to Bochner's theorem, a result from harmonic analysis about the relationship between a nonnegative function and its Fourier transform \cite{Bochner1,Bochner2}.

\begin{theorem}\label{thm:neg}
\textbf{Certifiable lower bound of Wigner negativity volume.}
Given some phase-space points $\Xi:=\{\vec{\xi}_k\}_{k=1}^N$, define the matrix $\mathbf{C}(\rho,\Xi)$ with elements
\begin{equation}
[\mathbf{C}(\rho,\Xi)]_{j,k} = \frac{1}{N}\tr[\rho D(\vec{\xi}_j-\vec{\xi}_k)].
\end{equation}
Then, $\mathcal{N}_C(\rho,\Xi) := \frac{1}{2} \tr[ |\mathbf{C}(\rho,\Xi)| - \mathbf{C}(\rho,\Xi)]$ is a lower bound for the Wigner negativity volume $\mathcal{N}_V(\rho)$ \cite{WignerNegativityVolume}:
\begin{equation}
\begin{aligned}
    \mathcal{N}_C(\rho,\Xi) \leq \mathcal{N}_V(\rho) := \frac{1}{2}\int\dd[2M]{\vec{\alpha}}\bqty\big{\abs{W_\rho(\vec{\alpha})}-W_\rho(\vec{\alpha})
    }.
\end{aligned}
\end{equation}
\end{theorem}
Theorem~\ref{thm:neg} thus provides a method to lower bound the Wigner logarithmic negativity $\log[2\mathcal{N}_V(\rho)-1]$, which is a commonly used monotone in the resource theories of non-Gaussianity and Wigner negativity \cite{NonGaussianResourceTheory1,NonGaussianResourceTheory2}.

In a previous study of Bochner's theorem as a negativity witness, it was shown that $\frac{1}{2}\operatorname{maxeig}[|\mathbf{C}(\rho,\Xi)|-\mathbf{C}(\rho,\Xi)]$ lower bounds $\mathcal{N}_{\tr}(\rho) := \min_{\sigma:W_\sigma \geq 0}\| \sigma - \rho \|_1$, which is a geometric measure of Wigner negativity based on the trace distance $\| \sigma - \rho \|_1 = \tr|\sigma-\rho|$  \cite{CharacteristicFunctionWitness}. However, the trace distance negativity has not been proven to be monotonic, and, for the examples studied here, $\frac{1}{2}\operatorname{maxeig}[|\mathbf{C}(\rho,\Xi)|-\mathbf{C}(\rho,\Xi)]$ is a tighter bound for $\mathcal{N}_V$ than $\mathcal{N}_{\tr}$.

Note that Bochner's theorem also guarantees the existence of some $\Xi$ such that $\mathcal{N}_C(\rho,\Xi) > 0$ for every $\rho$ with $\mathcal{N}_V(\rho) > 0$ \cite{Bochner1,Bochner2}. Hence, the possible choices of $\Xi$ define a collection of negativity witnesses, for which every Wigner-negative state can be detected by at least one of them. This collection therefore forms a complete family of Wigner negativity witnesses \cite{WitnessingWignerNegativity}.

The second theorem is a surprising connection between the negativity of the reduced Wigner function and the negativity of the partial transpose.
\begin{theorem}\label{thm:ent}
\textbf{Negativity of the reduced Wigner function implies partial transpose negativity.} Consider the equal bipartition of $M$ modes into $\vec{a}_{A} := \{a_m\}_{m=1}^{M/2}$ and $\vec{a}_{B} := \{a_m\}_{m=M/2+1}^{M}$. Denote the partial transpose of $\rho$ over the modes $\vec{a}_B$ as $\rho^{T_B}$, and the partial trace of $\rho$ over the collective modes $\vec{a}_- := (\vec{a}_A-\vec{a}_B)/\sqrt{2}$ as $\tr_{-}\rho$. Then, negativities in the Wigner function of $\tr_{-}\rho$ imply negativities in $\rho^{T_B}$. That is,
\begin{equation}\label{eq:partial-neg-ent}
    W_{\tr_{-}\rho}(\vec{\alpha}) \not\geq 0
    \implies
    \rho^{T_B} \not\succeq 0.
\end{equation}
\end{theorem}
Although an equal bipartition was assumed, Theorem~\ref{thm:ent} can be easily extended to an unequal bipartition $M_A > M_B$ by first performing a partial trace over $M_A-M_B$ modes in $\vec{a}_{A}$.

The two negativities were previously shown to be equivalent for a family of highly symmetric pure states \cite{NegativeEnt1}, while the non-negativity of the reduced Wigner function over a collective mode was known for general product states in the context of quantum convolutions \cite{Bochner2,Quantum-CLT}, and later extended to separable states via convexity \cite{TsirelsonEntanglement}. Theorem~\ref{thm:ent} thus subsumes the previous results about product and separable states \cite{Bochner2,Quantum-CLT,TsirelsonEntanglement}, while it is proven in Sec.~S2 of the Supplemental Material \cite{supplementaryMaterial} that Hudson's theorem and a restriction to highly symmetric pure states recover the result in Ref.~\cite{NegativeEnt1}.

The preceding theorems imply the following corollary.
\begin{corollary}\textbf{Certifiable lower bound of non-Gaussian entanglement.} Consider the equal bipartition of the $M$ total modes into $\vec{a}_{A}$ and $\vec{a}_{B}$. Choose some pairs of phase-space points $\Xi = \{(\vec{\xi}_k^A,\vec{\xi}_k^B)\}_{k=1}^N$, where every pair is symplectically related as
\begin{equation}\label{eq:symplectic-definition}
    \forall k : \pmqty{
        \vec{\xi}_k^A\\
        \vec{\xi}_k^{A*}
    } = \Lambda \pmqty{
        \vec{\xi}_k^B\\
        \vec{\xi}_k^{B*}
    }
\end{equation}
by the same $\Lambda$ for all $k$ such that $\Lambda^\dag \spmqty{\mathbbm{1}&0\\0&-\mathbbm{1}}\Lambda = \spmqty{\mathbbm{1}&0\\0&-\mathbbm{1}}$. Then, define the matrix $\mathbf{C}_2(\rho,\Xi)$ with elements 
\begin{equation}
[\mathbf{C}_2(\rho,\Xi)]_{j,k} = \frac{1}{N}\tr[\rho D_A(\vec{\xi}_j^A-\vec{\xi}_k^A)D_B(\vec{\xi}_j^B-\vec{\xi}_k^B)],
\end{equation}
where $D_A(\vec{\xi})$ ($D_B(\vec{\xi})$) is the displacement operator over the $\vec{a}_A$ ($\vec{a}_B$) modes. With this, the quantity $\mathcal{E}_C(\rho,\Xi) :=\frac{1}{2}\operatorname{maxeig}[|\mathbf{C}_2(\rho,\Xi)|-\mathbf{C}_2(\rho,\Xi)]$ lower bounds the following entanglement measures
\begin{equation}
\begin{aligned}
    \mathcal{E}_C(\rho,\Xi) \leq \mathcal{E}_{\operatorname{SEP}}(\rho) &:= \min_{\sigma \in \operatorname{SEP}} \|\sigma - \rho\|_1, \\
    \mathcal{E}_C(\rho,\Xi) \leq \mathcal{E}_{\operatorname{PPT}}(\rho) &:= \min_{\sigma \in \operatorname{PPT}} \|\sigma^{T_B} - \rho^{T_B}\|_1,
\end{aligned}
\end{equation}
where $\operatorname{SEP}$ and $\operatorname{PPT}$ are, respectively, the set of separable states and positive-partial-transpose states over the $\vec{a}_A$--$\vec{a}_B$ bipartition.
\end{corollary}
Here, $\mathcal{E}_{\operatorname{SEP}}(\rho)$ is a familiar geometric measure defined as the distance between $\rho$ and the set of separable states \cite{EntanglementWeak}. It is not itself an entanglement monotone \cite{EntanglementStrong}, although it can bound other measures that are \cite{SepEntanglementBound}. Meanwhile, $\mathcal{E}_{\operatorname{PPT}}(\rho)$ is a similar geometric measure conjectured to be equivalent to the partial transpose negativity $\tr|\rho^{T_B}| - 1$ \cite{LogNegativityConjecture}. This equivalence has been proven for a large class of states \cite{LogNegativityConjecture,LogNegativityKnown}. Unlike $\mathcal{E}_{\operatorname{SEP}}(\rho)$,  $\log \tr|\rho^{T_B}| $ is strongly monotonic \cite{LogNegativityMonotone}.

Clearly, $\mathcal{E}_C(\rho,\Xi) \leq \frac{1}{2}{\tr}[|\mathbf{C}_2(\rho,\Xi)|-\mathbf{C}_2(\rho,\Xi)]$, so $\mathcal{E}_C(\rho,\Xi)$ also lower bounds the Wigner negativity volume of $\rho$ by Theorem~\ref{thm:neg}. Therefore, $\mathcal{E}_C > 0$ simultaneously witnesses the Wigner negativity and entanglement of non-Gaussian entangled states, and provides quantitative lower bounds for some common measures of both.

\section{\label{sec:CDEW}Implementation of the Witness With Conditional Displacement Gates}
From the Corollary, the CD entanglement witness can be realized by carrying out the following steps:

(1) Choose $N$ phase-space pairs $\Xi = \{(\vec{\xi}_{k}^A,\vec{\xi}_{k}^B)\}_{k=1}^N$ where $\vec{\xi}_{k}^A$ and $\vec{\xi}_{k}^B$ are symplectically related as in Eq.~\eqref{eq:symplectic-definition}.

(2) Using CD gates and qubit measurements (see Fig.~\ref{fig:circuit}), measure $\langle{D_A(\vec{\xi}^A_j-\vec{\xi}^A_k)D_B(\vec{\xi}^B_j-\vec{\xi}^B_k)}\rangle \pm \delta_{j,k}$ for all $j < k$. Here, $\delta_{j,k}$ are the experimental error bars.

(3) Construct $\mathbf{C}_2$ with the matrix elements $[\mathbf{C}_2]_{j,k} = \langle{D_A(\vec{\xi}^A_j-\vec{\xi}^A_k)D_B(\vec{\xi}^B_j-\vec{\xi}^B_k)}\rangle/N$ for $j<k$ from the previous step, while the other elements are given by $[\mathbf{C}_2]_{j,j} = 1$ and $[\mathbf{C}_2]_{j,k} = [\mathbf{C}_2]_{k,j}^*$ for $j>k$.

(4) Calculate $\mathcal{E}_C := \frac{1}{2}\operatorname{maxeig}[|\mathbf{C}_2|-\mathbf{C}_2]$ and $\delta:=\max_j\sum_{k\neq j} \delta_{j,k}/N$. Within the reported experimental uncertainty, the system is entangled when $\mathcal{E}_C > \delta$, and $\mathcal{E}_C \pm \delta$ is a lower bound to $\mathcal{N}_V$, $\mathcal{N}_{\tr}$, $\mathcal{E}_{\operatorname{SEP}}$, and $\mathcal{E}_{\operatorname{PPT}}$.

\begin{figure}
    \centering
    \includegraphics{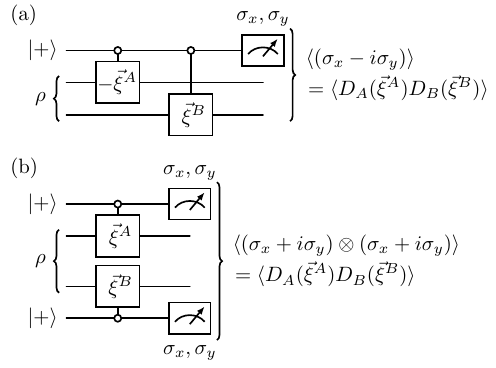}
    \caption{Circuit to measure $\langle{D_A(\vec{\xi}_j^A-\vec{\xi}_k^A)D_B(\vec{\xi}_j^B-\vec{\xi}_k^B)}\rangle$, with (a) one or (b) two auxiliary qubits. The former requires fewer qubits and qubit measurements, while the latter requires only local gates and local measurements on the $A$ and $B$ subsystems. CD gates $U_{\text{CD}}(\vec{\xi})$ are drawn with open circles $\circ$ at the control qubit input and target displacements as boxes containing the argument $\vec{\xi}$. Note the negative signs in (a) due to an accumulated phase from chaining two CD gates.}
    \label{fig:circuit}
\end{figure}

The proof that $\pm \delta$ is the propagated experimental error is given in Sec.~S3 of the Supplemental Material \cite{supplementaryMaterial}. The circuit that measures $[\mathbf{C}_2]_{j,k}$, which requires a single instance of the CD gate for each subsystem, and up to two auxiliary qubits, is shown in Fig.~\ref{fig:circuit}. Notice the freedom to either use collective measurements, which reduces the number of auxiliary qubits needed, or measurements local to each subsystem, which would be required to witness entanglement between remote subsystems.

In total, up to $n_q N(N-1)$ measurement settings are needed, where $n_q$ is the number of auxiliary qubits. This takes into account the qubit measurements ($\sigma_x$, $\sigma_y$ for each qubit), and the CD gate settings ($\vec{\xi}_{j}^{\mu}-\vec{\xi}_k^{\mu}$ for $j < k$). The number of measurement settings can be reduced with symmetries in the choice of $\Xi$, as will be shown in the following examples. 

If $n_s$ samples are taken for each measurement setting, a na\"ive estimate with the standard error gives $\delta_{j,k} \sim 1/\sqrt{n_s}$ for each matrix element of $\mathbf{C}_2$ \cite{statistics}. Since $\mathcal{E}_C$ depends on an eigenvalue of $\mathbf{C}_2$, non-Gaussian entanglement is certified when the propagated error $\delta=\max_j\sum_{k\neq j} \delta_{j,k}/N\sim 1/\sqrt{n_s}$ satisfies $\delta < \mathcal{E}_C$, as derived in Sec.~S3 of the Supplemental Material \cite{supplementaryMaterial}. Therefore, at least $n_s \sim \mathcal{E}_C^{-2}$ samples per setting, or a total of $\sim n_q N^2\mathcal{E}_C^{-2}$ samples, are required. While this scaling applies in the asymptotic limit due to the central limit theorem, better estimates can be obtained by including characterized sources of experimental error, or by performing finite statistical analysis on the witness \cite{ent-finite}.

\subsection{\label{sec:choice-of-points}Choice of Phase-space Points}
It is desirable to choose $\Xi$ such that $\mathcal{E}_C$ is large for the target state $\rho$, as that would ensure that the observed value of $\mathcal{E}_C$ would still be positive within experimental error bars, even with imperfect state preparation and measurements. To achieve this, since $\mathcal{E}_C = \frac{1}{2}\operatorname{maxeig}[|\mathbf{C}_2|-\mathbf{C}_2] = \max(-\lambda_{-},0)$ for the minimum eigenvalue $\lambda_{-}$ of $\mathbf{C}_2(\rho,\Xi)$, an initial heuristic choice $\Xi = \{(\vec{\xi}_k^A,\vec{\xi}_k^B)\}_{k=1}^N$ can be improved via gradient descent on $\lambda_{-}$ \cite{gradientDescent}. This requires the the partial derivative of $\lambda_{-}$ with respect to $\vec{\xi}_k^A$, which is \cite{eigenvalueDerivative}
\begin{equation}
\frac{\partial\lambda_{-}}{\partial\vec{\xi}^{A*}_{k}} = \frac{\vec{v}^\dag ( \partial \mathbf{C}_2 / \partial \vec{\xi}^{A*}_{k} ) \vec{v}}{\vec{v}^\dag\vec{v}},
\end{equation}
where $\vec{v}$ is the minimum eigenvector such that $\mathbf{C}_2\vec{v} = \lambda_{-}\vec{v}$, and $ \partial \mathbf{C}_2 / \partial \vec{\xi}^{A*}_{k} $ can be found by taking the partial derivative of each matrix element of $\mathbf{C}_2$.

There are also scenarios where an optimal choice of $\Xi$ for a state $\rho$ was already found from a prior optimization, and one would now be interested in witnessing the non-Gaussian entanglement of another state $\widetilde{\rho}$ symplectically related to $\rho$. That is, $\widetilde{\rho} = U_+ U_- U_A U_B \rho  U_B^\dag U_A^\dag U_-^\dag U_+^\dag$ via Gaussian unitaries $U_{\mu}$ for $\mu \in \{A,B,+,-\}$, where \cite{GaussianQMreview}
\begin{equation}\label{eq:sympletic-transformed}
    \pmqty{
        U_{\mu}^\dag \vec{a}_\mu U_\mu \\
        U_{\mu}^\dag \vec{a}_\mu^\dag U_\mu
    } = \Lambda_\mu \pmqty{
        \vec{a}_\mu\\
        \vec{a}_\mu^\dag
    } + \pmqty{\vec{\alpha}_{0}^\mu \\ \vec{\alpha}_{0}^{\mu *}}
\end{equation}
for symplectic $\Lambda_\mu$ such that $\Lambda_\mu^\dag \spmqty{\mathbbm{1}&0\\0&-\mathbbm{1}}\Lambda_\mu = \spmqty{\mathbbm{1}&0\\0&-\mathbbm{1}}$, and
\begin{equation}
    \pmqty{
        \vec{a}_\pm\\
        \vec{a}_\pm^\dag
    } := \frac{1}{\sqrt{2}}\pmqty{\vec{a}_A\\\vec{a}_A^\dag} \pm \frac{1}{\sqrt{2}}\,\Lambda\pmqty{\vec{a}_B\\\vec{a}_B^\dag}
\end{equation}
for the $\Lambda$ in Eq.~\eqref{eq:symplectic-definition} that relates $\vec{\xi}_k^A$ with $\vec{\xi}_k^B$. It is proven in Sec.~S4 of the Supplemental Material \cite{supplementaryMaterial} that $\mathcal{E}_C(\widetilde{\rho},\widetilde{\Xi}) = \mathcal{E}_C(\rho,\Xi)$, where $\widetilde{\Xi} = \{(\widetilde{\xi}_{k}^A,\widetilde{\xi}_{k}^B)\}_{k=1}^N$ such that
\begin{equation}\label{eq:sympletic-transform-1}
    \pmqty{
        \widetilde{\xi}_{k}^A \\
        \widetilde{\xi}_{k}^{A*}
    } = \Lambda_+ \Lambda_A \pmqty{
        \vec{\xi}_{k}^A\\
        \vec{\xi}_{k}^{A *}
    },\,\pmqty{
        \widetilde{\xi}_{k}^B \\
        \widetilde{\xi}_{k}^{B*}
    } = \Lambda_+ \Lambda_B \pmqty{
        \vec{\xi}_{k}^B\\
        \vec{\xi}_{k}^{B *}
    }.
\end{equation}
Therefore, the CD witness can also be used to detect $\widetilde{\rho}$ with the same expected value of $\mathcal{E}_C$ by simply choosing the phase-space pairs $\widetilde{\Xi}$ instead of $\Xi$.

\subsection{\label{sec:examples}Application on Example States}
\begin{figure}
    \centering
    \includegraphics{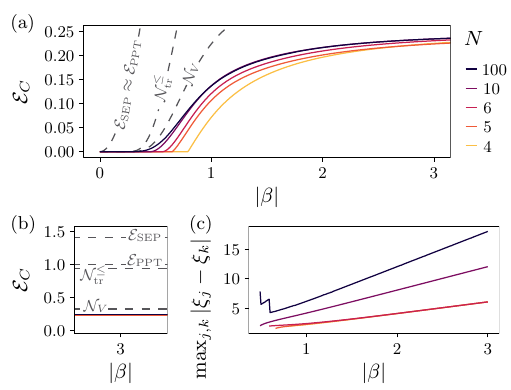}
    \caption{Application of the CD witness on entangled cats. (a) $\mathcal{E}_C$ against $|\beta|$, with all Wigner negativity and entanglement measures also shown. $\mathcal{E}_{\text{SEP}} \approx \mathcal{E}_{\text{PPT}}$ only for $|\beta| \ll 1$. Non-Gaussian entanglement of every cat state with $|\beta| \gtrsim 3/4$ can be detected. (b) Plot expanded in the region around $|\beta|=3$. Notice that $\mathcal{E}_C$ is a tighter bound for $\mathcal{N}_{V}$ than for $\mathcal{N}_{\tr}$, and $\mathcal{E}_{\operatorname{PPT}}$ than for $\mathcal{E}_{\operatorname{SEP}}$. (c) Maximum magnitude of CD displacements, which scales linearly with $|\beta|$. The nonlinear variations at small $|\beta|$ are due to the small amount of Wigner negativities at those values, which makes them sensitive to the numerical precision of the computer when optimizing heuristically.}
    \label{fig:entangledCats}
\end{figure}

As an application of the CD witness, consider entangled cats that take the form $\ket{\operatorname{Cat}_2(\beta)} \propto \ket{\beta,\beta} + \ket{-\beta,-\beta}$, where $\ket{\beta_1,\beta_2}$ are the usual coherent states.

 $\mathcal{E}[\ket{\operatorname{Cat}_2(\beta)},\Xi]$ is plotted against $|\beta|$ in Fig.~\ref{fig:entangledCats}(a), and compared in detail to the other measures in the neighborhood of $|\beta| = 3$ in Fig.~\ref{fig:entangledCats}(b). Note that only a lower bound $\mathcal{N}_{\tr}^{\leq} \leq \mathcal{N}_{\tr}$ could be obtained for the trace distance negativity, see Sec.~S6 of the Supplemental Material for details \cite{supplementaryMaterial}.

Every cat state with $|\beta| \gtrsim 3/4$ can be detected with $N \geq 4$, while a larger $N$ is required to witness cat states of smaller magnitudes. Notice also that $\mathcal{E}_C$ is a tighter lower bound for $\mathcal{N}_V$ than $\mathcal{N}_{\tr}$, and $\mathcal{E}_{\operatorname{PPT}}$ than $\mathcal{E}_{\operatorname{SEP}}$. However, since $\mathcal{E}_C$ lower bounds all measures simultaneously, and $\mathcal{N}_V$ is much smaller than the other measures, $\mathcal{E}_C$ only loosely bounds $\mathcal{N}_{\tr}$, $\mathcal{E}_{\operatorname{PPT}}$, and $\mathcal{E}_{\operatorname{SEP}}$.

For $N=4$, the phase-space points used in Fig.~\ref{fig:entangledCats} are $\Xi = \{(\mu \xi_\nu,\mu \xi_\nu) : \mu,\nu \in \pm\}$, where $\xi_{\pm} := \beta \pm i\pi/(8\beta^*)$. The only displacements needed are $\xi_j-\xi_k \in \{ 2\beta, i\pi/(4\beta^*), 2\xi_{+}, 2\xi_{-} \}$, so the non-Gaussian entanglement of $\ket{\operatorname{Cat}_2(\beta)}$ with $|\beta| \gtrsim 3/4$ can be witnessed with only four points of the characteristic function.

Applications of the CD witness on other example states are given in Sec.~S5 of the Supplemental Material, including entangled Fock states, photon-subtracted two-mode squeezed vacua, and noisy states affected by photon loss \cite{supplementaryMaterial}. It is shown there that four points of the characteristic function are sufficient to detect most of these states, and that the witness is fairly robust against noise with or without any characterization of the noise present.

\section{Conclusion}
In this Letter, I have proposed a non-Gaussian entanglement witness that uses only CD gates and qubit measurements. It does not require full or partial tomography and is particularly suitable for weakly dispersive cQED systems where quadrature measurements and parity gates are unavailable or prohibitively difficult. The joint characteristic function measurements required for the CD witness are already possible in present-day systems---such measurements have just recently been demonstrated by \citet{CNODgate}.

Furthermore, the expected value of the witness is a lower bound to the Wigner negativity volume and a geometric measure of entanglement conjectured to be the partial transpose negativity, thus providing an experimentally accessible method to lower bound quantities related to these measures.

Common non-Gaussian entangled states can be witnessed by measuring as few as four points of the characteristic function. The witness is also easily adaptable for states related by symplectic transformations, which is immediately applicable to recently introduced noise mitigation techniques \cite{SqueezedCatsTheory1,SqueezedCatsTheory2,SqueezedCats}.

A few open questions remain. One is the possibility of improving the lower bound $\mathcal{E}_C$ by using the trace instead of the maximum eigenvalue in Theorem~\ref{thm:ent}. While it has been observed that $\frac{1}{2}\tr[|\mathbf{C}_2|-\mathbf{C}_2] \leq \mathcal{N}_V \leq \mathcal{E}_{\operatorname{PPT}},\mathcal{E}_{\operatorname{SEP}}$ for all example states, I have not managed to prove this in this Letter.

Another open question is the relationship between $\mathcal{E}_C$ and the partial transpose negativity. While the former is a lower bound for the latter for all example states, the conjecture about $\mathcal{E}_{\operatorname{PPT}}$ is still unproven for general states. As the CD witness is not a fidelity witness, usual techniques that relate entanglement witnesses with quantitative estimates of the partial transpose negativity give trivial lower bounds \cite{quantify-EW-1,quantify-EW-2}.

Lastly, a possible future research direction is the extension of these results to the study of magic and mana in discrete Wigner functions \cite{magic}. While Theorem~\ref{thm:neg} should carry over naturally to define a lower bound of the discrete analog of Wigner negativity volume, Theorem~\ref{thm:ent} might pose some difficulty due to its dependence on collective modes. This is because collective modes in continuous variable systems are defined with a balanced beam splitter, but their discrete analog might not always exist, as beam splitters are only defined for specific transmissivities and dimensions in discrete variable systems \cite{discrete-beamsplitter}.

\section*{Acknowledgments}
\begin{acknowledgments}
I thank Jonathan Schwinger for introducing me to the problem of witnessing entanglement with CD gates, and an anonymous referee for pointing out a stronger result about Theorem~\ref{thm:neg}. I am also grateful for helpful discussions with Valerio Scarani, Simone Gasparinetti, and Axel Eriksson.

This Letter is supported by the National Research Foundation, Singapore, and A*STAR under its CQT Bridging Grant.  
\end{acknowledgments}

\bibliography{refs}

\clearpage

\phantomsection\addcontentsline{toc}{part}{Supplementary Material for: Certifiable Lower Bounds of Wigner Negativity Volume and Non-Gaussian Entanglement With Conditional Displacement Gates}

\title{Supplementary Material for: Certifiable Lower Bounds of Wigner Negativity Volume and Non-Gaussian Entanglement With Conditional Displacement Gates}

\maketitle

\setcounter{section}{0}
\setcounter{figure}{0}
\setcounter{equation}{0}
\renewcommand{\thesection}{S\arabic{section}}
\renewcommand{\thefigure}{S\arabic{figure}}
\renewcommand{\theequation}{S\arabic{equation}}

\section{\label{apd:proofs}Proofs of Theorems}
\begin{theorem2}\label{thm:neg2}
\textbf{Certifiable Lower Bound of Wigner Negativity Volume.}
Given some phase-space points $\Xi:=\{\vec{\xi}_k\}_{k=1}^N$, define the matrix $\mathbf{C}(\rho,\Xi)$ with elements
\begin{equation}
[\mathbf{C}(\rho,\Xi)]_{j,k} = \frac{1}{N}\tr[\rho D(\vec{\xi}_j-\vec{\xi}_k)].
\end{equation}
Then, $\mathcal{N}_C(\rho,\Xi) := \frac{1}{2} \tr[ |\mathbf{C}(\rho,\Xi)| - \mathbf{C}(\rho,\Xi)]$ is a lower bound for the Wigner negativity volume $\mathcal{N}_V(\rho)$ \cite{WignerNegativityVolume}:
\begin{equation}
\begin{aligned}
    \mathcal{N}_C(\rho,\Xi) \leq \mathcal{N}_V(\rho) := \frac{1}{2}\int\dd[2M]{\vec{\alpha}}\bqty\big{\abs{W_\rho(\vec{\alpha})}-W_\rho(\vec{\alpha})
    }.
\end{aligned}
\end{equation}
\end{theorem2}
\begin{proof}
Using the fact that ${\tr}[\rho D(\vec{\xi})]$ is the Fourier transform of $W_\rho(\vec{\alpha})$, the matrix elements of $\mathbf{C}(\rho,\Xi)$ are
\begin{equation}
    [\mathbf{C}(\rho,\Xi)]_{j,k} = \frac{1}{N}\int\dd[2M]{\vec{\alpha}} 
    W_\rho(\vec{\alpha}) e^{\vec{\xi}_j\wedge\vec{\alpha}} e^{-\vec{\xi}_k\wedge\vec{\alpha}}.
\end{equation}
Recall that $\vec{\xi}\wedge\vec{\alpha} =\sum_{m=1}^M \alpha_m \xi_m^* - \alpha_m^*\xi_m$, which is antisymmetric with respect to conjugation $(\vec{\xi}\wedge\vec{\alpha})^* = -\vec{\xi}\wedge\vec{\alpha}$. Let $\vec{v}(\vec{\alpha}) := (e^{\vec{\xi}_1\wedge\vec{\alpha}},e^{\vec{\xi}_2\wedge\vec{\alpha}} ,\dots,e^{\vec{\xi}_N\wedge\vec{\alpha}})^T/\sqrt{N}$. Then,
\begin{equation*}
    \mathbf{C}(\rho,\Xi) = \int\dd[2M]{\vec{\alpha}} W_\rho(\vec{\alpha}) \vec{v}(\alpha) \vec{v}^\dag(\alpha)
\end{equation*}
with normalization $\lvert \vec{v}(\vec{\alpha}) \rvert^2 = 1$, such that $\vec{v}(\alpha) \vec{v}^\dag(\alpha) \succeq 0$ is a rank-one projector.

Meanwhile, the Wigner function can be written as the difference $W_{\rho}(\vec{\alpha}) = W_{\rho}^{+}(\vec{\alpha}) - W_{\rho}^{-}(\vec{\alpha})$ of two nonnegative functions $W_{\rho}^{\pm}(\vec{\alpha}) := \frac{1}{2}[|W_{\rho}(\vec{\alpha})| \pm W_{\rho}(\vec{\alpha})] \geq 0$. As such,
\begin{equation}\label{eq:order-of-C}
\begin{aligned}
    -\mathbf{C}(\rho,\Xi)
    &= \int\dd[2M]{\vec{\alpha}}
            W_\rho^-(\vec{\alpha}) \vec{v}(\alpha) \vec{v}^\dag(\alpha)  \\[-1.5ex]
            &\qquad\qquad{}-{}\int\dd[2M]{\vec{\alpha}}
            W_\rho^+(\vec{\alpha}) \vec{v}(\alpha) \vec{v}^\dag(\alpha) \\
    &\preceq \int\dd[2M]{\vec{\alpha}}
            W_\rho^-(\vec{\alpha}) \vec{v}(\alpha) \vec{v}^\dag(\alpha),
\end{aligned}
\end{equation}
where the second line comes from the fact that the first line is the difference of two positive semidefinite matrices. Next, note that
\begin{equation}
\begin{aligned}
    \mathcal{N}_C(\rho,\Xi) &= \frac{1}{2}\Bqty{\tr|\mathbf{C}(\rho,\Xi)| - \tr[\mathbf{C}(\rho,\Xi)]} \\
    &= \tr{\Pi_- \bqty{-\mathbf{C}(\rho,\Xi)}},
\end{aligned}
\end{equation}
where $\Pi_-$ is the projector onto the negative eigenspace of $\mathbf{C}(\rho,\Xi)$. Substituting this back into Eq.~\eqref{eq:order-of-C},
\begin{equation}
\begin{aligned}
    \mathcal{N}_C(\rho,\Xi) &\leq \tr[\Pi_- \int\dd[2M]{\vec{\alpha}}
            W_\rho^-(\vec{\alpha}) \vec{v}(\alpha) \vec{v}^\dag(\alpha) ] \\
    &=\int\dd[2M]{\vec{\alpha}}
            W_\rho^-(\vec{\alpha})  \underbrace{\vec{v}^\dag(\alpha) \Pi_- \vec{v}(\alpha)}_{\leq 1} \\
    &\leq \int\dd[2M]{\vec{\alpha}}
            W_\rho^-(\vec{\alpha}) = \mathcal{N}_V(\rho),
\end{aligned}
\end{equation}
which completes the proof.
\end{proof}

\begin{theorem2}\label{thm:ent2}
\textbf{Negativity of the Reduced Wigner Function Implies Partial Transpose Negativity.} Consider the equal bipartition of $M$ modes into $\vec{a}_{A} := \{a_m\}_{m=1}^{M/2}$ and $\vec{a}_{B} := \{a_m\}_{m=M/2+1}^{M}$. Denote the partial transpose of $\rho$ over the modes $\vec{a}_B$ as $\rho^{T_B}$, and the partial trace of $\rho$ over the collective modes $\vec{a}_- := (\vec{a}_A-\vec{a}_B)/\sqrt{2}$ as $\tr_{-}\rho$. Then, negativities in the Wigner function of $\tr_{-}\rho$ imply negativities in $\rho^{T_B}$. That is,
\begin{equation}\label{eq:partial-neg-ent2}
    W_{\tr_{-}\rho}(\vec{\alpha}) \not\geq 0
    \implies
    \rho^{T_B} \not\succeq 0.
\end{equation}
\end{theorem2}
\begin{proof}
First, define the phase-space coordinates $\vec{\alpha}_\pm := (\vec{\alpha}_A \pm \vec{\alpha}_B)/\sqrt{2}$ of the collective modes $\vec{a}_{\pm} := (\vec{a}_A \pm \vec{a}_B)/\sqrt{2}$. The partial trace over $\vec{a}_-$ is given by the integral of the Wigner function over the $\vec{\alpha}_-$ coordinate \cite{wignerReduced}
\begin{equation}
\begin{aligned}
    &W_{\tr_{-}\rho}(\vec{\alpha}_+) \\
    &\quad{}={} \int\dd[M]{\vec{\alpha}_-} W_\rho\bqty{\frac{1}{\sqrt{2}}(\vec{\alpha}_++\vec{\alpha}_-),\frac{1}{\sqrt{2}}(\vec{\alpha}_+-\vec{\alpha}_-)} \\
    &\quad{}={} 2^M \int\dd[M]{\vec{\gamma}_-} W_\rho\pqty{
        \vec{\gamma}_+ +\vec{\gamma}_-,
        \vec{\gamma}_+-\vec{\gamma}_-
    },
\end{aligned}
\end{equation}
where the change of variables $\vec{\gamma}_{\pm} := \vec{\alpha}_\pm/\sqrt{2}$ was carried out in the last step.

In terms of the Wigner function $W_O(\vec{\alpha}_A,\vec{\alpha}_B)$ of an operator $O$, the Wigner function of $O^{T_B}$ is given by $W_{O^{T_B}}(\vec{\alpha}_A,\vec{\alpha}_B) = W_{O}(\vec{\alpha}_A,\vec{\alpha}_B^*)$ \cite{WignerTranspose}. As such,
\begin{equation}\label{eq:EPR-Wigner-Overlap}
\begin{aligned}
    &\int\dd[M]{\vec{\gamma}_-} W_\rho(\vec{\gamma}_+ - \vec{\gamma}_-,\vec{\gamma}_+ + \vec{\gamma}_-)  \\
    &= \int\dd[M]{\vec{\gamma}_-} W_{\rho^{T_B}}(\vec{\gamma}_+ - \vec{\gamma}_-,\vec{\gamma}_+^{*} + \vec{\gamma}_-^{*}) \\
    &= \int\dd[M]{\vec{\alpha}_A}\int\dd[M]{\vec{\alpha}_B}
    \delta(\vec{\alpha}_A+\vec{\alpha}_B^*-2\vec{\gamma}_+) W_{\rho^{T_B}}(\vec{\alpha}_A,\vec{\alpha}_B) \\
    &\propto \int\dd[M]{\vec{\alpha}_A}\int\dd[M]{\vec{\alpha}_B}
    W_{\Phi}(\vec{\alpha}_A,\vec{\alpha}_B) W_{\rho^{T_B}}(\vec{\alpha}_A,\vec{\alpha}_B).
\end{aligned}
\end{equation}
Here, I have identified $\delta(\vec{\alpha}_A+\vec{\alpha}_B^*-2\vec{\gamma}_+)$ as the Wigner function of the Einstein--Poldosky--Rosen (EPR) state $\Phi$ with center-of-mass position $\sqrt{2}\Re(\vec{\gamma}_+)$ and relative momentum $2\sqrt{2}\Im(\vec{\gamma}_+)$ \cite{EPR-Bell}. Although there are subtleties involved with the normalization of the EPR state \footnote{While interpretations of Eq.~\eqref{eq:EPR-Wigner-Overlap} as a probability can be problematic due to the unnormalizability of the EPR state \cite{EPR-Unnormalized}, there are no issues here as the only property needed is its positive semidefiniteness.}\nocite{EPR-Unnormalized}, $\Phi \succeq 0$ is positive semidefinite for any choice of normalization. Therefore,
\begin{equation}
\begin{gathered}
    \rho^{T_B} \succeq 0 
    \implies \tr(\Phi\rho^{T_B}) \geq 0  \\
    \implies \int\dd[M]{\vec{\gamma}_-} W_\rho(\vec{\gamma}_+ - \vec{\gamma}_-,\vec{\gamma}_+ + \vec{\gamma}_-) \geq 0 \\
    \implies W_{\tr_{-}\rho}(\vec{\alpha}_+) \geq 0.
\end{gathered}
\end{equation}
Taking the converse gives Eq.~\eqref{eq:partial-neg-ent}, as desired.
\end{proof}

\begin{corollary2}\label{col:neg-ent2}\textbf{Certifiable Lower Bound of Non-Gaussian Entanglement.} Consider the equal bipartition of the $M$ total modes into $\vec{a}_{A}$ and $\vec{a}_{B}$. Choose some pairs of phase-space points $\Xi = \{(\vec{\xi}_k^A,\vec{\xi}_k^B)\}_{k=1}^N$, where every pair is symplectically related as
\begin{equation}\label{eq:symplectic-definition2}
    \forall k : \pmqty{
        \vec{\xi}_k^A\\
        \vec{\xi}_k^{A*}
    } = \Lambda \pmqty{
        \vec{\xi}_k^B\\
        \vec{\xi}_k^{B*}
    }
\end{equation}
by the same $\Lambda$ for all $k$ such that $\Lambda^\dag \spmqty{\mathbbm{1}&0\\0&-\mathbbm{1}}\Lambda = \spmqty{\mathbbm{1}&0\\0&-\mathbbm{1}}$. Then, define the matrix $\mathbf{C}_2(\rho,\Xi)$ with elements 
\begin{equation}
[\mathbf{C}_2(\rho,\Xi)]_{j,k} = \frac{1}{N}\tr[\rho D_A(\vec{\xi}_j^A-\vec{\xi}_k^A)D_B(\vec{\xi}_j^B-\vec{\xi}_k^B)],
\end{equation}
where $D_A(\vec{\xi})$ ($D_B(\vec{\xi})$) is the displacement operator over the $\vec{a}_A$ ($\vec{a}_B$) modes. With this, the quantity $\mathcal{E}_C(\rho,\Xi) :=\frac{1}{2}\operatorname{maxeig}[|\mathbf{C}_2(\rho,\Xi)|-\mathbf{C}_2(\rho,\Xi)]$ lower bounds the following entanglement measures
\begin{equation}
\begin{aligned}
    \mathcal{E}_C(\rho,\Xi) \leq \mathcal{E}_{\operatorname{SEP}}(\rho) &:= \min_{\sigma \in \operatorname{SEP}} \|\sigma - \rho\|_1, \\
    \mathcal{E}_C(\rho,\Xi) \leq \mathcal{E}_{\operatorname{PPT}}(\rho) &:= \min_{\sigma \in \operatorname{PPT}} \|\sigma^{T_B} - \rho^{T_B}\|_1,
\end{aligned}
\end{equation}
where $\operatorname{SEP}$ and $\operatorname{PPT}$ are respectively the set of separable states and positive-partial-transpose states over the $\vec{a}_A$--$\vec{a}_B$ bipartition.
\end{corollary2}
\begin{proof}
The corollary is trivially true when $\mathcal{E}_C(\rho,\Xi) = 0$, so only the $\mathcal{E}_C(\rho,\Xi) > 0$ case needs to be considered. Assume first that $\Lambda = \mathbbm{1}$ in Eq.~\eqref{eq:symplectic-definition2}. That is, $\vec{\xi}_k^A = \vec{\xi}_k^B =: \vec{\xi}_k$ for all $k$. Since $\mathcal{E}_C(\rho,\Xi) > 0$, this means that there is a normalized vector $\vec{v} = (v_k)_{k=1}^N$ such that $\vec{v}^\dagger \mathbf{C}_2(\rho,\Xi) \vec{v} = -\frac{1}{2}\operatorname{maxeig}[|\mathbf{C}_2(\rho,\Xi)|-\mathbf{C}_2(\rho,\Xi)] = - \mathcal{E}_C(\rho,\Xi)$. Now, consider the operator
\begin{equation}
    C_2 := \frac{1}{N}\sum_{j,k=1}^N v_j^* v_k D_A(\vec{\xi}_j-\vec{\xi}_k)D_B(\vec{\xi}_j-\vec{\xi}_k),
\end{equation}
defined such that $\forall\sigma:\tr(\sigma C_2) = \vec{v}^\dagger \mathbf{C}_2(\sigma,\Xi) \vec{v}$. Notice that both $C_2$ and its partial transpose  \cite{WignerTranspose}
\begin{equation}
C_2^{T_B} = \frac{1}{N}\sum_{j,k=1}^N v_j^* v_k D_A(\vec{\xi}_j-\vec{\xi}_k)D_B(\vec{\xi}_k^*-\vec{\xi}_j^*)
\end{equation}
take the form $\sum_{j,k} v_j^*v_k D(\vec{\xi}_{j,k}')/N$ for some argument of $\vec{\xi}_{j,k}'$ with $\vec{\xi}_{j,j}' = 0$. As such, writing $C_2^{\square}$ in place of either $C_2$ or $C_2^{T_B}$,
\begin{equation}
    \lambda_{-}^{\square} \leq \tr(\sigma C_2^{\square}) \leq \lambda_{+}^{\square}
\end{equation}
where $\lambda_{-}^{\square}$ and $\lambda_{+}^{\square}$ are the smallest and largest eigenvalues of the matrix with elements ${\tr}[\sigma D(\vec{\xi}_{j,k}')]/N$. By the Gershgorin circle theorem  \cite{MatrixAnalysis},
\begin{equation}
\abs{\lambda_{\pm}-\frac{1}{N}} \leq \max_j \frac{1}{N}\sum_{k\neq j} \abs{\tr[\sigma D(\vec{\xi}_{j,k}')]} \leq 1 - \frac{1}{N},
\end{equation}
hence $|\lambda_{\pm}| \leq 1$. It can also be directly verified that $C_2^\square$ is Hermitian, so $\lambda_\pm$ must be real. Therefore,
\begin{equation}
    \forall \sigma : -1 \leq \tr(\sigma C_2^{\square}) \leq 1 \implies -\mathbbm{1} \preceq C_2,C_2^{T_B} \preceq \mathbbm{1}.
\end{equation}
Notice also that the elements of $\mathbf{C}_2(\sigma,\Xi)$ are
\begin{equation}
\begin{aligned}
{[\mathbf{C}_2(\sigma,\Xi)]}_{j,k} &= \tr[\sigma \; D_A(\vec{\xi}_j-\vec{\xi}_k)D_B(\vec{\xi}_j-\vec{\xi}_k)] \\
&= \tr{\sigma \; \exp[
    \pqty{\sqrt{2}\vec{\xi}_j - \sqrt{2}\vec{\xi}_k} \wedge \frac{\vec{a}_1 + \vec{a}_2}{\sqrt{2}}
] } \\
&= \tr[\sigma \; D_+(\sqrt{2}\vec{\xi}_j-\sqrt{2}\vec{\xi}_k)] \\
&= \tr[\tr_{-}\sigma \; D_+(\sqrt{2}\vec{\xi}_j-\sqrt{2}\vec{\xi}_k)],
\end{aligned}
\end{equation}
where $D_+(\alpha)$ is the displacement operator defined on the $\vec{a}_+$ modes, $\tr_{-}\sigma$ is the partial trace of $\sigma$ over the $\vec{a}_-$ modes, and $\vec{a}_\pm := (\vec{a}_1 \pm \vec{a}_2)/\sqrt{2}$. 

By Theorem~\ref{thm:ent2}, for any $\sigma : \sigma^{T_B} \succeq 0$, the Wigner function of $\tr_{-}\sigma$ will be nonnegative, for which Theorem~\ref{thm:neg2} further implies that $\tr(\sigma C_2) = \vec{v}^\dagger \mathbf{C}_2(\sigma,\Xi) \vec{v} \geq 0$. Then, with $\|\sigma-\rho\|_1 = \max_{-\mathbbm{1} \preceq V \preceq \mathbbm{1}} \tr[V(\sigma-\rho)]$ \cite{MikeAndIke},
\begin{equation}\label{eq:lower-bound-entanglement}
\begin{aligned}
    \mathcal{E}_{\operatorname{PPT}}(\rho) &= \min_{\sigma \in \operatorname{PPT}}\|\sigma^{T_B} - \rho^{T_B}\|_1 \\
    &= \min_{\sigma \in \operatorname{PPT}}
    \max_{-\mathbbm{1}\preceq V \preceq \mathbbm{1}} \tr[V(\sigma^{T_B}-\rho^{T_B})] \\
    &\geq \min_{\sigma \in \operatorname{PPT}} \tr(\sigma^{T_B} C_2^{T_B}) - \tr(\rho^{T_B} C_2^{T_B}) \\
    &\geq 0 - \tr(\rho C_2) = \mathcal{E}_C(\rho,\Xi).
\end{aligned}
\end{equation}
Analogous steps give $\mathcal{E}_{\operatorname{SEP}}(\rho) \geq \mathcal{E}_C(\rho,\Xi)$, thus completing the proof for the case when $\Lambda = \mathbbm{1}$ in Eq.~\eqref{eq:symplectic-definition2}.

To extend this to $\Lambda \neq \mathbbm{1}$, use the fact that there exists a Gaussian unitary $U$ such that \cite{GaussianQMreview}
\begin{equation}
    U^\dag \vec{a}_A U = \vec{a}_A, \quad \pmqty{
        U^\dag \vec{a}_B U \\
        U^\dag \vec{a}_B^\dag U
    } = \Lambda \pmqty{
        \vec{a}_B \\
        \vec{a}_B^\dag
    },
\end{equation}
where $U$ is separable over the $\vec{a}_A$--$\vec{a}_B$ bipartition. Let $\Xi_A := \{(\vec{\xi}_k^A,\vec{\xi}_k^A)\}_{k=1}^N$. From the preceding proof,
\begin{equation}\label{eq:corollary-lower-bounds}
\begin{aligned}
    \mathcal{E}_C(U\rho U^\dag,\Xi_A) &\leq \mathcal{E}_{\operatorname{SEP}}(U\rho U^\dag) =  \mathcal{E}_{\operatorname{SEP}}(\rho), \\
    \mathcal{E}_C(U\rho U^\dag,\Xi_A) &\leq \mathcal{E}_{\operatorname{PPT}}(U\rho U^\dag) =  \mathcal{E}_{\operatorname{PPT}}(\rho),
\end{aligned}
\end{equation}
where the entanglement measures are unchanged under the separable unitary transformation. Finally, since
\begin{equation}
\begin{aligned}
    &{[\mathbf{C}_2(U\rho U^\dag,\Xi_A)]}_{j,k} \\
     &\quad{}={} \tr[U\rho U^\dag D_A(\vec{\xi}_j^A-\vec{\xi}_k^A)D_B(\vec{\xi}_j^A-\vec{\xi}_k^A)] \\
     &\quad{}={} \tr[\rho D_A(\vec{\xi}_j^A-\vec{\xi}_k^A) \, U^\dag  D_B(\vec{\xi}_j^A-\vec{\xi}_k^A) U]\\
     &\quad{}={} \tr[\rho D_A(\vec{\xi}_j^A-\vec{\xi}_k^A) D_B(\vec{\xi}_j^B-\vec{\xi}_k^B)] \\
     &\quad{}={} {[\mathbf{C}_2(\rho,\Xi)]}_{j,k},
\end{aligned}
\end{equation}
with $\Xi = \{(\vec{\xi}_k^A,\vec{\xi}_k^B)\}_{k=1}^N$ and $\vec{\xi}_k^B$ as defined in Eq.~\eqref{eq:symplectic-definition2}, this implies that $\mathcal{E}_C(U\rho U^\dag,\Xi_A) = \mathcal{E}_C(\rho,\Xi)$, which substituted back into Eq.~\eqref{eq:corollary-lower-bounds} completes the proof.
\end{proof}

\section{\label{apd:s-wave}Equivalence of Wigner Negativity and Entanglement for Pure Hyperspherical \texorpdfstring{$s$}{s} Waves}
\citet{NegativeEnt1} studied hyperspherical $s$ waves of the form $\ket{\Psi} = \int\dd{\vec{x}}\Psi(\abs{\vec{x}})\ket{\vec{x}}$ that depend only on the magnitude of the total position vector. Note that $\vec{x}$ includes the positions of multiple particles, hence their nomenclature of ``\emph{hyperspherical}''. Such states are invariant under the transformation $\vec{x} \to R\vec{x}$ for every orthogonal matrix $R^TR = \mathbbm{1}$, which imply that their Wigner functions similarly satisfy $W_{\ket{\Psi}}(R\vec{\alpha}) = W_{\ket{\Psi}}(\vec{\alpha})$.

Note that the proof for the first implication appeared in the original work by \citet{NegativeEnt1}, repeated here for convenience. The proof for the second implication is unique to this work.
\begin{theorem}
    \textbf{\citet{NegativeEnt1}}. A pure hyperspherical $s$ wave is entangled if and only if its Wigner function has negativities.
\end{theorem}
\begin{proof}
    (\emph{Entanglement Implies Wigner Negativity}). Assume that $W_{\ket{\Psi}}$ is nonnegative. The multimode extension of Hudson's theorem \cite{Hudson1} by \citet{Hudson2} implies that $W_{\ket{\Psi}}(\vec{\alpha})$ must be Gaussian. The only Gaussian invariant under all rotations is the Wigner function for the separable state $\ket{\xi}^{\otimes M}$, where $\ket{\xi} = \exp(-\xi\wedge a^2/2)\ket{0}$ is the squeezed state. Therefore, taking the converse implies that if $\ket{\Psi}$ is entangled, then it must have negativities in its Wigner function.

    (\emph{Wigner Negativity Implies Entanglement}). Take any bipartition of the $M$ modes into $\vec{a}_A$ and $\vec{a}_B$ with $M_A \geq M_B$. Assume $\ket{\Psi} = \ket{\Psi_A} \otimes \ket{\Psi_B}$. If $M_A > M_B$, further split $\vec{a}_{A}$ into modes $\vec{a}_{A'}$ and $\vec{a}_{\Delta}$ such that $M_{A'}=M_B$, then take a partial trace over $\vec{a}_{\Delta}$ so that $\tr_{\Delta}\ketbra{\Psi} = \rho_{A'} \otimes \ketbra{\Psi_B}$. Define the collective modes $\vec{a}_{\pm} := (\vec{a}_{A'} \pm \vec{a}_B)/\sqrt{2}$, which can be transformed to $\vec{a}_{A'}$ or $\vec{a}_B$ by a rotation. Therefore, by rotational invariance,
    \begin{equation}\label{eq:Wigner-s-wave}
    \begin{aligned}
        \tr_- \pqty{\tr_{\Delta}\ketbra{\Psi}} &=
        \tr_{{A'}} \pqty{\tr_{\Delta}\ketbra{\Psi}} \hspace{-0.85em} &&= \tr_{B} \pqty{\tr_{\Delta}\ketbra{\Psi}} \\
        &=  \ketbra{\Psi_B}\hspace{-0.85em} &&=\rho_{A'} .
    \end{aligned}
    \end{equation}
    Firstly, since $\rho_{A'} = \tr_\Delta \ketbra{\Psi_A} = \ketbra{\Psi_B}$, this means that $\ket{\Psi_A} = \ket{\Psi_B} \otimes \ket{\Psi_\Delta}$ for some $\ket{\Psi_\Delta}$ defined on the modes $\vec{a}_{\Delta}$. Secondly, since $\tr_{\Delta}\ketbra{\Psi}$ is separable, Theorem~\ref{thm:ent} implies that the Wigner function of $\tr_-(\tr_{\Delta}\ketbra{\Psi})$, hence that of every state in Eq.~\eqref{eq:Wigner-s-wave}, must be nonnegative. Repeating this argument with a bipartition over $\vec{a}_{A'+B}$ and $\vec{a}_{\Delta}$ imply that the Wigner function of $\ket{\Psi_\Delta}$ is also nonnegative, so the Wigner function of $\ket{\Psi} = \ket{\Psi_B} \otimes \ket{\Psi_\Delta} \otimes \ket{\Psi_B}$ is a product of nonnegative Wigner functions, which must itself be nonnegative. Taking the converse implies that if the Wigner function of $\ket{\Psi}$ has negativities, then it must be entangled.
\end{proof}

\section{\label{apd:errorPropagation}Propagation of Experimental Error Bars}
To consider how error bars in the experimental data will be propagated to $\mathcal{E}_C$, take that the data is reported to be $|\hat{d}_{j,k} - d_{j,k}| \leq \delta_{j,k}$ within a certain confidence level, where $\hat{d}_{j,k}$ ($d_{j,k}$) is the estimator (estimand) of $\langle D_A(\vec{\xi}_j^A-\vec{\xi}_k^A)D_B(\vec{\xi}_j^B-\vec{\xi}_k^B) \rangle$. Hence, the matrix $\mathbf{C}_2$ can be written as $\widehat{\mathbf{C}}_2 = \mathbf{C}_2 + \Delta$, where $[\widehat{\mathbf{C}}_2]_{j,k} = \hat{d}_{j,k}/N$, $[{\mathbf{C}}_2]_{j,k} = d_{j,k}/N$, and $[\Delta]_{j,k}$ can take any value within the radius $|[\Delta]_{j,k}| \leq \delta_{j,k}/N$. Writing the minimum eigenvalues of $(\widehat{\mathbf{C}}_2, \mathbf{C}_2, \Delta)$ as $(\hat{\lambda}_{-},\lambda_{-},\delta_{-})$, and the maximum eigenvalue of $\Delta$ as $\delta_{+}$, Weyl's inequalities state that \cite{MatrixAnalysis}\begin{equation}\label{eq:error-range}
    -\hat{\lambda}_- + \delta_- \leq -\lambda_-
    \leq - \hat{\lambda}_- + \delta_+.
\end{equation}
Using Gershgorin circle theorem on $\Delta$ \cite{MatrixAnalysis},
\begin{equation}
\abs{\max \delta_+}, \abs{\min\delta_-} \leq \max_{j}\sum_{k \neq j}\frac{\delta_{j,k}}{N}  =:  \delta.
\end{equation}
Hence, the widest range of Eq.~\eqref{eq:error-range} consistent with the experimental data and the reported error is $|\hat{\lambda}_{-} - \lambda_-| \leq \delta$. Lastly, notice that $\max(-\lambda_-,0) = \frac{1}{2}\operatorname{maxeig}[|\mathbf{C}_2|-\mathbf{C}_2] = \mathcal{E}_C$, and similarly for $\widehat{\mathcal{E}}_C$. Therefore, within the given confidence level, the system is certain to be entangled when $\widehat{\mathcal{E}}_C > \delta$, and the experimental error bars are propagated to the computed lower bound as $\mathcal{E}_C = \widehat{\mathcal{E}}_C \pm \delta$.

\begin{widetext}
\section{\label{apd:gaussianUnitaries}Choice of Phase-space Points for State Related By Gaussian Unitaries}
Given $\widetilde{\rho} = U_+ U_- U_A U_B \rho  U_B^\dag U_A^\dag U_-^\dag U_+^\dag$ via Gaussian unitaries $U_{\mu}$ for $\mu \in \{A,B,+,-\}$, where
\begin{equation}\label{eq:apd-sympletic-transformed}
    \pmqty{
        U_{\mu}^\dag \vec{a}_\mu U_\mu \\
        U_{\mu}^\dag \vec{a}_\mu^\dag U_\mu
    } = \Lambda_\mu \pmqty{
        \vec{a}_\mu\\
        \vec{a}_\mu^\dag
    } + \pmqty{\vec{\alpha}_{0}^\mu \\ \vec{\alpha}_{0}^{\mu *}}
\end{equation}
for some symplectic $\Lambda_\mu$, and
\begin{equation}
    \pmqty{
        \vec{a}_\pm\\
        \vec{a}_\pm^\dag
    } := \frac{1}{\sqrt{2}}\pmqty{\vec{a}_A\\\vec{a}_A^\dag} \pm \frac{1}{\sqrt{2}}\,\Lambda\pmqty{\vec{a}_B\\\vec{a}_B^\dag}
\end{equation}
for the symplectic $\Lambda$ that relates $\vec{\xi}_k^A$ with $\vec{\xi}_k^B$,
\begin{equation}\label{eq:apd-sympletic-transform-1}
\begin{aligned}
    {}[\mathbf{C}_2(\rho,\Xi)]_{j,k}
    &= \tr[D_A(\vec{\xi}_j^A-\vec{\xi}_k^A)D_B(\vec{\xi}_j^B-\vec{\xi}_k^B)\;\rho] \\
    &= \tr[
        D_A(\vec{\xi}_j^A-\vec{\xi}_k^A)D_B(\vec{\xi}_j^B-\vec{\xi}_k^B)\;
        U_B^\dag U_A^\dag U_-^\dag U_+^\dag\widetilde{\rho} U_+ U_- U_A U_B  
    ] \\
    &= \tr_+\!\Bqty{
        U_+ \bqty{
            U_A D_A(\vec{\xi}_j^A-\vec{\xi}_k^A)U_A^\dag \;
            U_B D_B(\vec{\xi}_j^B-\vec{\xi}_k^B)U_B^\dag
        } U_+^\dag \;
        \tr_-\!\pqty{U_-^\dag \widetilde{\rho} U_-}
    } \\
    &= \tr[
        D_A(\widetilde{\xi}_{j}^A-\widetilde{\xi}_{k}^A)
        D_B(\widetilde{\xi}_{j}^B-\widetilde{\xi}_{k}^B)
        \widetilde{\rho}
    ]  \; e^{\widetilde{\alpha}^A_{0} \wedge (\widetilde{\xi}_{j}^A-\widetilde{\xi}_{k}^A)}
    e^{\widetilde{\alpha}^B_{0} \wedge (\widetilde{\xi}_{j}^B-\widetilde{\xi}_{k}^B)}\\
    &= [\mathbf{C}_2(\widetilde{\rho},\widetilde{\Xi})]_{j,k} \; e^{\widetilde{\alpha}^A_{0} \wedge (\widetilde{\xi}_{j}^A-\widetilde{\xi}_{k}^A)}
    e^{\widetilde{\alpha}^B_{0} \wedge (\widetilde{\xi}_{j}^B-\widetilde{\xi}_{k}^B)},
\end{aligned}
\end{equation}
where $\tr_\pm$ is the partial trace over the $\vec{a}_\pm$ modes. Here, I have also defined $\widetilde{\Xi} := \{(\widetilde{\xi}_{k}^A,\widetilde{\xi}_{k}^B)\}_{k=1}^N$, where $\widetilde{\xi}_{k}^\mu$ and $\widetilde{\alpha}_{0}^\mu$ for $\mu \in \{A,B\}$ are given by
\begin{equation}\label{eq:apd-sympletic-transform-2}
\begin{aligned}
    \pmqty{
        \widetilde{\xi}_{k}^\mu \\
        \widetilde{\xi}_{k}^{\mu*}
    } &= \Lambda_+ \Lambda_\mu \pmqty{
        \vec{\xi}_{k}^\mu\\
        \vec{\xi}_{k}^{\mu *}
    },&
    \pmqty{
        \widetilde{\alpha}_{0}^\mu \\
        \widetilde{\alpha}_{0}^{\mu *}
    } &= \Lambda_+ \pmqty{
        \vec{\alpha}_{0}^{\mu*}\\
        \vec{\alpha}_{0}^{\mu*}
    } + \pmqty{
        \vec{\alpha}_{0}^+ \\
        \vec{\alpha}_{0}^{+ *}
    }.
\end{aligned}
\end{equation}
The phase factors in the last line of Eq.~\eqref{eq:apd-sympletic-transform-1} act as a unitary transformation on $\mathbf{C}_2(\widetilde{\rho},\widetilde{\Xi})$, which leaves the minimum eigenvalue unchanged. Therefore, $\mathcal{E}_C(\rho,\Xi) = \mathcal{E}_C(\widetilde{\rho},\widetilde{\Xi})$.
\end{widetext}

\section{\label{apd:otherStates}Application of CD Witness on Other Example States}
In this section, the CD witness is applied to some other examples of non-Gaussian entangled states. The heuristically-optimized values of $\mathcal{E}_C(\rho,\Xi)$ are plotted for these states, and compared against the other measures of Wigner negativity and entanglement. Details on obtaining the other measures for the example states can be found in Appendix~\ref{apd:otherMeasures}.

\subsection{Entangled Single-photon Fock States}
\begin{figure}
    \centering
    \includegraphics{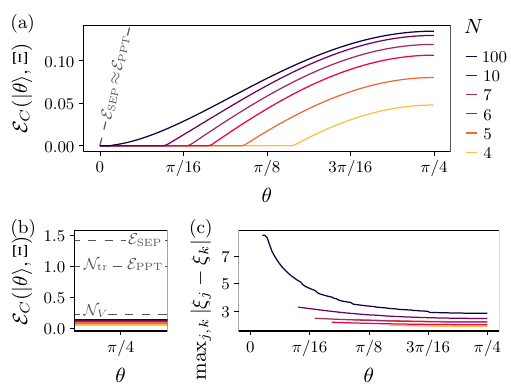}
    \caption{Entangled single-photon Fock states. (a) Expectation value $\mathcal{E}_C$ of the CD witness against $\theta$. $\mathcal{E}_{\text{SEP}} \approx \mathcal{E}_{\text{PPT}} \approx 2\theta$ for $\theta \ll 1$ is plotted as a gray dashed line, while both $\mathcal{N}_{V}$ and $\mathcal{N}_{\tr}$ are outside plot range. Non-Gaussian entanglement is detected for $\theta \gtrsim 0.01\pi$ with $N=100$, and for $\theta \gtrsim 0.146\pi$ with $N=4$. (b) Same plot in the neighborhood of $\theta = \pi/4$, with the vertical axis expanded to show all negativity and entanglement measures. Notice in particular that $\mathcal{E}_C$ is a tighter bound for $\mathcal{N}_{V}$ than for $\mathcal{N}_{\tr}$, and $\mathcal{E}_{\operatorname{PPT}}$ than for $\mathcal{E}_{\operatorname{SEP}}$. (c) Maximum magnitude of the displacements required to implement the CD witness, which scales approximately as $\sim (\theta-\pi/4)^2$.}
    \label{fig:bellType}
\end{figure}
The Fock state $\ket{1,0}$ passed through a beamsplitter gate $B(\theta):=e^{\theta(a_1^\dag a_2 - a_1 a_2^\dag)}$ is given by
\begin{equation}
\ket{\theta} := B(\theta)\ket{1,0} = \cos\theta \ket{1,0} + \sin\theta\ket{0,1},
\end{equation}
which is entangled for any $\theta > 0$ and maximally entangled at $\theta = \pi/4$.

Applying the CD witness on $\ket{\theta}$, the expectation value $\mathcal{E}_C$ is plotted in Fig.~\ref{fig:bellType}(a), and compared in detail to the Wigner negativity and entanglement measures in the neighborhood of the maximally-entangled state $\ket{\theta=\pi/4}$ in Fig.~\ref{fig:bellType}(b). The measures $\mathcal{N}_V$, $\mathcal{N}_{\tr}$, $\mathcal{E}_{\operatorname{PPT}}$, and $\mathcal{E}_{\operatorname{SEP}}$ were all found analytically.

An immediate observation is that more measurements are required to detect the state at smaller angles: $\mathcal{E}_C > 0$ for $\theta \gtrsim 0.146\pi$ with $N=4$, while $\mathcal{E}_C > 0$ for $\theta \gtrsim 0.01\pi$ with $N=100$. Notice also that $\mathcal{E}_C$ is a tighter lower bound for $\mathcal{N}_V$ than $\mathcal{N}_{\tr}$, and for $\mathcal{E}_{\operatorname{PPT}}$ than $\mathcal{E}_{\operatorname{SEP}}$.

Figure~\ref{fig:bellType}(c) shows the maximum magnitude of the displacements needed when performing the CD gates during the implementation of this witness. The dependence on $\theta$ is approximately quadratic, with the least displacement required when the state is maximally entangled.

For $N=4$, the phase-space points used in Fig.~\ref{fig:bellType} are $\Xi = \{(\pm \Re\xi_0,\pm \Re\xi_0),(\pm i\Im\xi_0,\pm i\Im\xi_0)\}$, where
\begin{equation}
\begin{aligned}
\xi_0 {}={} & i\Bqty{\frac{813}{1217} + \frac{\cos\!\bqty{\frac{1249}{171}\pqty{\frac{\pi}{4} - \theta} + \frac{2179}{215}\pqty{\frac{\pi}{4} - \theta}^4}}{1313}}\\
&\quad{}+{}\frac{2531}{2745} + \frac{453}{2083}\pqty{\frac{\pi}{4}-\theta}^2.
\end{aligned}
\end{equation}
For this choice of $\Xi$, the only displacements needed are $\xi_j-\xi_k \in \{2\Re[\xi_0],i 2\Im[\xi_0],2\xi_0,2\xi_0^*\}$, as the other matrix elements of $\mathbf{C}_2$ can be obtained with $\langle{D(-\vec{\xi})}\rangle = \langle{D(\vec{\xi})}\rangle^*$. Thus, the non-Gaussian entanglement of $\ket{\theta}$ with $\theta \gtrsim 0.146\pi$ can be witnessed by measuring only four points of the characteristic function.

\subsection{Photon-subtracted Two-mode Squeezed Vacua}
\begin{figure}
    \centering
    \includegraphics{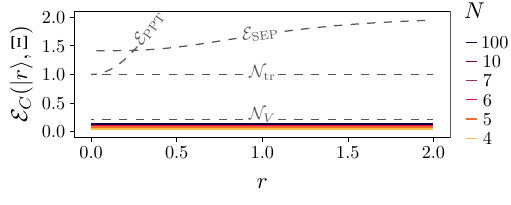}
    \caption{Photon-subtracted two-mode squeezed vacua. The expected value of the CD witness and Wigner negativity measures are the same as for Fig.~\ref{fig:bellType} at $\theta=\pi/4$, as this state is symplectically related to the maximally-entangled Fock state. Meanwhile, both entanglement measures increase with the amount of squeezing. Unlike before, $\mathcal{E}_C$ is not always a tighter bound for $\mathcal{E}_{\operatorname{PPT}}$, since $\mathcal{E}_{\operatorname{PPT}} > \mathcal{E}_{\operatorname{SEP}}$ for $r \gtrsim 0.25$. Generally, $\mathcal{E}_C$ is a rather loose bound for the entanglement measures, but it is also the case that $\mathcal{E}_C > 0$ for all $r$, so the CD witness can certify the non-Gaussian entanglement of all photon-subtracted two-mode squeezed vacua as defined in Eq.~\eqref{eq:PSTMSS}.}
    \label{fig:TMSS}
\end{figure}

Photon-subtracted two-mode squeezed vacua, which can be prepared by postselecting on a photon-detection event on squeezed vacuum states, are of the form \cite{PSTMSS}
\begin{equation}\label{eq:PSTMSS}
\begin{aligned}
    \ket{r} &\propto B(\pi/4)\, a_1 \,S_1(-r) S_2(r)\ket{0,0} \\
    &\propto (a_1+a_2)\ket{\text{TMSV}(r)} \\
    &\propto S_+(-r) S_-(r)\ket{\theta=\pi/4},
\end{aligned}
\end{equation}
where $B(\theta)$ is the beamsplitter gate from before, $S_\mu(r) = \exp[(r/2)(a_\mu^2-a_\mu^{\dagger 2})]$ is the squeeze operator for the mode $a_\mu \in \{a_1,a_2,a_+,a_-\}$ with $a_\pm \propto a_1 \pm a_2$, and $\ket{\text{TMSV}(r)} \propto \sum_{n=0}^\infty \tanh^n(r)\ket{n,n}$ is the two-mode squeezed vacuum.

The last line of Eq.~\eqref{eq:PSTMSS} shows that $\ket{r}$ is symplectically related to the maximally-entangled Fock state from the previous section, so the Wigner negativity measures $\mathcal{N}_V$ and $\mathcal{N}_{\tr}$ for every $\ket{r}$ are the same as those for $\ket{\theta=\pi/4}$. Furthermore, using Eq.~\eqref{eq:sympletic-transform-1}, it is possible to obtain the same expected value $\mathcal{E}_C(\ket{r},\Xi_r) = \mathcal{E}_C(\ket{\theta=\pi/4},\Xi)$ of the CD witness with
\begin{equation}
    \Xi_r := \{ \Re[\xi]e^{-r} + i \Im[\xi]e^r : \xi \in \Xi \}.
\end{equation}
These quantities are plotted together with the entanglement measures in Fig.~\ref{fig:TMSS}. Entanglement increases with squeezing parameter $r$, with $\mathcal{E}_{\operatorname{PPT}} \leq \mathcal{E}_{\operatorname{SEP}}$ for $r \lesssim 0.25$ and $\mathcal{E}_{\operatorname{PPT}} \geq \mathcal{E}_{\operatorname{SEP}}$ for $r \gtrsim 0.25$. Meanwhile, the value of $\mathcal{E}_C > 0$ is the same for every $r$. On one hand, this means that $\mathcal{E}_C$ is a very loose lower bound for the entanglement measures, especially as $r$ increases. On the other hand, this also means that the CD witness can witness the entanglement of a photon-subtracted two-mode squeezed vacuum state with any amount of squeezing.

\begin{figure}
    \centering
    \includegraphics{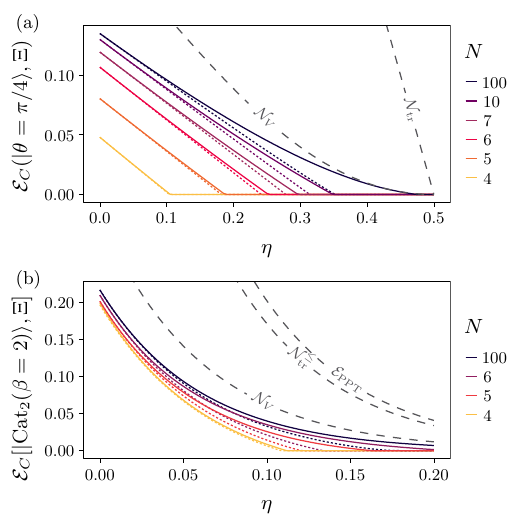}
    \caption{Performance of the CD witness in the presence of photon loss for (a) the maximally-entangled Fock state $\ket{\theta = \pi/4}$ and (b) the entangled cat state $\ket{\operatorname{Cat}_2(\beta=2)}$. $\mathcal{E}_{\text{SEP}}$ is outside the plotted range. Two scenarios are shown: where the phase-space points $\Xi$ of the lossless case is na\"ively used even in the presence of noise (dotted traces), and where $\Xi$ have been further maximized for the noisy state (solid traces). The latter scenario captures the situation where the noise has been fully characterized.}
    \label{fig:noisyStates}
\end{figure}

\subsection{\label{sec:robustness}States In The Presence of Photon Loss}
Consider the noise channel
\begin{equation}
\mathcal{L}_\eta[\bullet] := \tr_{\vec{b}}\Bqty{ U_\eta \bqty{\bullet \otimes \ketbra{0}^{\otimes M} } U^\dagger_\eta },
\end{equation}
where $\vec{b}$ are the modes of the environment and $U_\eta \vec{a} U_\eta^\dagger = \sqrt{1-\eta}\vec{a} + \sqrt{\eta} \vec{b}$. This describes the loss of photons to the environment with some noise parameter $\eta$, where $\eta = 0$ is the lossless case.

The performance of the witness in the presence of photon loss is plotted in Fig.~\ref{fig:noisyStates} for the example states $\ket{\theta=\pi/4}$ and $\ket{\operatorname{Cat}_2(\beta=2)}$. Two scenarios with different choices of $\Xi$ are considered: where $\Xi$ of the lossless state is na\"ively used (dotted traces), and where $\Xi$ has been further maximized for the noisy state (solid traces). The former scenario captures the case where noise is present but uncharacterized, in contrast to the latter scenario where the noise model is known.

As seen in Fig.~\ref{fig:noisyStates}, the CD witness is fairly robust against noise: even without further optimization of $\Xi$, it detects the non-Gaussian entanglement of $\ket{\theta=\pi/4}$ up to a photon loss of $\eta \lesssim 0.35$, and that of $\ket{\operatorname{Cat}_2(\beta=2)}$ up to $\eta \lesssim 0.19$. An even wider range of noisy states can be detected with $\Xi$ optimized for the characterized noise: in fact, $\mathcal{E}_C$ even saturates the Wigner negativity volume for $\ket{\theta=\pi/4}$ with photon loss $0.4 \lesssim \eta < 0.5$.

\section{\label{apd:otherMeasures}Negativity and Entanglement Measures of Example States}
For convenience, the measures to be found are
\begin{equation}
\begin{aligned}
    \mathcal{N}_V(\rho) &= \frac{1}{2}\int\dd[2m]{\vec{\alpha}}\pqty{|W_\rho(\vec{\alpha})| - W_\rho(\vec{\alpha})} \\
    \mathcal{N}_{\tr}(\rho) &= \min_{\sigma:W_\sigma \geq 0} \|\sigma-\rho\|_1 \\
    \mathcal{E}_{\operatorname{SEP}}(\rho) &= \min_{\sigma\in\operatorname{SEP}} \|\sigma-\rho\|_1 = \min_{\ket{a,b}\in\operatorname{SEP}} \|\ketbra{a,b}-\rho\|_1 \\
    \mathcal{E}_{\operatorname{PPT}}(\rho) &= \min_{\sigma\in\operatorname{PPT}} \|\sigma^{T_B}-\rho^{T_B}\|_1,
\end{aligned}
\end{equation}
where the minimization over the convex set $\operatorname{SEP}$ is replaced by the minimization over its extremal points, which are the pure separable states.

For pure states with the Schmidt decomposition $\ket{\psi} = \sum_{k} \sqrt{p_k} |{\phi_k^A}\rangle\otimes|{\phi_k^B}\rangle$, 
\begin{equation}\label{eq:ent-measure}
\begin{aligned}
    \mathcal{E}_{\operatorname{SEP}}(\ket{\psi}) &=  2\sqrt{1 - \max_{\ket{a,b}\in\operatorname{SEP}} \abs{\braket{a,b}{\psi}}^2} = 2\sqrt{1 - \max_k p_k} \\
    \mathcal{E}_{\operatorname{PPT}}(\ket{\psi}) &= \pqty{\sum_k \sqrt{p_k}}^2 - 1.
\end{aligned}
\end{equation}
The last line comes from the closed-form expression of $\mathcal{E}_{\operatorname{PPT}}(\ket{\psi}) = |\ketbra{\psi}^{T_B}|-1$ by \citet{LogNegativityKnown}. Meanwhile, the first line is obtained from rewriting the trace distance in terms of the fidelity as $\|\ketbra{\phi}-\ketbra{\psi}\|_1 = 2(1-\abs{\braket{\phi}{\psi}}^2)^{1/2}$ \cite{MikeAndIke}. Then, with 
\begin{equation}
\begin{aligned}
\abs{\braket{a,b}{\psi}} &= \abs{\sum_{k} \sqrt{p_k} \braket{a}{\phi_k^A}\braket{b}{\phi_k^B}} \\
&\leq \sum_{k} \sqrt{p_k} \abs{\braket{a}{\phi_k^A}}\abs{\braket{b}{\phi_k^B}}.
\end{aligned}
\end{equation}
It can be seen that $|\braket{a,b}{\psi}|$ is maximized by some $\ket{a} = \sum_k \sqrt{q_k} \lvert{\phi_k^A}\rangle$ and $\ket{b} = \sum_k \sqrt{r_k} \lvert{\phi_k^B}\rangle$ such that $\langle{a}\lvert{\phi_k^A}\rangle$ and $\langle{b}\lvert{\phi_k^B}\rangle$ are nonnegative. Finally, from the inequality
\begin{equation}
\begin{aligned}
|\braket{a,b}{\psi}| &= \sum_{k} \sqrt{p_kq_kr_k} \\
&\leq \sqrt{\max_kp_k} \sum_{k'}\sqrt{q_{k'} r_{k'}} \\
&\leq  \sqrt{\max_k p_k},
\end{aligned}
\end{equation}
this gives
\begin{equation}
    \mathcal{E}_{\operatorname{SEP}}(\ket{\psi}) \geq 2\sqrt{1 - \max_k p_k},
\end{equation}
which is saturated by the choice $\ket{a,b} = |\phi_{k}^A\rangle \otimes |{\phi_{k}^B}\rangle$, proving the equality.

Therefore, both entanglement measures can be found for pure states by simply substituting their Schmidt coefficients into Eq.~\eqref{eq:ent-measure}. This covers the entanglement measures for the example states, so the rest of the appendix will only discuss the Wigner negativity measures.

\subsection{Entangled Cat States}
The Wigner negativity volume of $\ket{\operatorname{Cat}_2(\beta)} \propto \ket{\beta,\beta} + \ket{-\beta,-\beta}$ was found using numerical integration. Meanwhile, a lower bound on the trace distance negativity for a single-mode state $\rho_1$ is \cite{WitnessingWignerNegativity}
\begin{equation}
    \min_{\sigma:W_\sigma \geq 0}\| \sigma - \rho_1 \|_1 \geq - \frac{\pi}{2} W_{\rho_1}(\alpha).
\end{equation}
Noting that the entangled cat state can be symplectically transformed with a beamsplitter transformation to $\ket{\operatorname{Cat}_1(\sqrt{2}\beta)} \propto (|{\sqrt{2}\beta}\rangle+|{-\sqrt{2}\beta}\rangle)\otimes\ket{0}$,
\begin{equation}
\begin{aligned}
    \mathcal{N}_{\tr}\bqty{\ket{\operatorname{Cat}_2(\beta)}}
    &\geq \mathcal{N}_{\tr} \bqty{\ket{\operatorname{Cat}_1(\sqrt{2}\beta)}}\\
    &\geq - \frac{\pi}{2} W_{\ket{\operatorname{Cat}_1(\sqrt{2}\beta)}}(\alpha).
\end{aligned}
\end{equation}
The lower bound reported in the main text is $\mathcal{N}^{\leq}_{\tr}(\rho) := \max[0, -(\pi/2) \min_\alpha W_{\ket{\operatorname{Cat}_1(\sqrt{2}\beta)}}(\alpha)]$, where the minimization was done numerically.

\subsection{Entangled Single-photon Fock States and Photon-subtracted Two-mode Squeezed Vacuum}
Since Wigner negativity is invariant under symplectic transformations, the negativity volume can directly integrated to find $ \mathcal{N}_V(\ket{\theta}) = \mathcal{N}_V(\ket{r}) = \mathcal{N}_V(\ket{1,0}) = 2e^{-1/2}-1$ \cite{WignerNegativityVolume}.

Meanwhile, the geometric measure of Wigner negativity is lower-bounded by $\mathcal{N}_{\tr}(\ket{1,0}) \geq \mathcal{N}_{\tr}(\ket{1}) \geq 2(1 - \max_{\sigma:W_\sigma \geq 0}\bra{1}\sigma\ket{1})$ by the contractive property of the trace distance and its relationship to the fidelity \cite{MikeAndIke}. The maximum fidelity is known to be $\max_{\sigma:W_\sigma \geq 0}\bra{1}\sigma\ket{1} = 1/2$ \cite{WitnessingWignerNegativity}. Hence, $\mathcal{N}_{\tr}(\ket{1,0}) \geq 1$. This bound is tight, as it is saturated by the Wigner-nonnegative state
\begin{equation}
    \sigma = \frac{1}{2}\pqty{\ketbra{0} + \ketbra{1}} \otimes \ketbra{0}.
\end{equation}
Therefore, $\mathcal{N}_V(\ket{\theta}) = \mathcal{N}_V(\ket{r}) = \mathcal{N}_{\tr}(\ket{1,0}) = 1$.
\end{document}